\documentclass[10pt]{article}
\usepackage[pdftex]{hyperref}
\usepackage{color}
\newif\ifFOCS \FOCSfalse

\usepackage{verbatim}
\usepackage[cmex10]{amsmath}
\usepackage{amssymb}
\usepackage{amsthm}
\renewcommand{\Pr}{\mathbf{Pr}}

\newcommand{\opt}{\mathsf{opt}}

\newcommand{\E}{\mathbf{E}}

\newcommand{\R}{\mathbb{R}}
\newcommand{\SampleTree}{\mathtt{SampleTree}}

\newcommand{\eps}{\varepsilon}
\DeclareMathOperator{\sgn}{sgn}
\DeclareMathOperator{\lmax}{lmax}

\ifFOCS
\newcommand{\FOCSspace}[1]{}
\newcommand{\mybox}[1]{\medskip \noindent\fbox{\parbox{0.97\columnwidth}{#1}}\medskip}
\else
\newcommand{\FOCSspace}[1]{}
\newcommand{\mybox}[1]{\medskip \noindent\fbox{\parbox{0.98\textwidth}{#1}}\medskip}
\fi
\newcommand{\grant}{Research supported in part by a Microsoft Graduate Fellowship and in part by a NSF Graduate Fellowship}
\newtheorem{theorem}{Theorem}[section]

\newtheorem{fact}[theorem]{Fact}
\newtheorem{defn}[theorem]{Definition}
\newtheorem{lemma}[theorem]{Lemma}
\newtheorem{example}[theorem]{Example}

\begin{document}
\ifFOCS
\else
\title{Nearly Maximum Flows in Nearly Linear Time}
\author{Jonah Sherman\thanks{\grant}\\University of California, Berkeley}
\date{\today (preliminary draft)}
\fi
\maketitle
\thispagestyle{empty}
\begin{abstract}
We introduce a new approach to the maximum flow problem in undirected, capacitated graphs using $\alpha$-\emph{congestion-approximators}: easy-to-compute functions that approximate the congestion required to route single-commodity demands in a graph to within a factor of $\alpha$.  Our algorithm maintains an arbitrary flow that may have some residual excess and deficits, while taking steps to minimize a potential function measuring the congestion of the current flow plus an over-estimate of the congestion required to route the residual demand.  Since the residual term over-estimates, the descent process gradually moves the contribution to our potential function from the residual term to the congestion term, eventually achieving a flow routing the desired demands with nearly minimal congestion after $\tilde{O}(\alpha\eps^{-2}\log^2 n)$ iterations. Our approach is similar in spirit to that used by Spielman and Teng (STOC 2004) for solving Laplacian systems, and we summarize our approach as trying to do for $\ell_\infty$-flows what they do for $\ell_2$-flows.

Together with a nearly linear time construction of a $n^{o(1)}$-congestion-approximator, we obtain $1+\eps$-optimal single-commodity flows undirected graphs in time $m^{1+o(1)}\eps^{-2}$, yielding the fastest known algorithm for that problem.  Our requirements of a congestion-approximator are quite low, suggesting even faster and simpler algorithms for certain classes of graphs.  For example, an $\alpha$-competitive oblivious routing tree meets our definition, \emph{even without knowing how to route the tree back in the graph}.  For graphs of conductance $\phi$, a trivial $\phi^{-1}$-congestion-approximator gives an extremely simple algorithm for finding $1+\eps$-optimal-flows in time $\tilde{O}(m\phi^{-1})$.
\end{abstract}

\section{Introduction}
The maximum flow problem and its dual, the minimum cut problem are fundamental combinatorial optimization problems with a wide variety of applications.
In the well-known maximum $s-t$ flow problem we are given a graph $G$ with edge capacities $c_e$, and aim to route as much flow as possible from $s$ to $t$ while restricting the magnitude of the flow on each edge to its capacity.  We will prefer instead to think in terms of the equivalent problem of routing a single unit of flow from $s$ to $t$ while minimizing the maximum congestion $|f_e/c_e|$ on any edge; clearly the minimum congestion for unit flow is equal to one divided by the maximum flow of congestion one.  Once formulated that way, we need no longer restrict ourselves to $s-t$ flows;
given a \emph{demand vector} $b \in \R^n$ specifying the excess desired at each vertex, we aim to find a flow $f \in \R^m$ with divergence equal to $b$ that minimizes the maximum congestion $|f_e/c_e|$.\footnote{In fact the $s-t$ case is no less general, since one could always add a new vertices $s$ and $t$, connect each $v$ to $s$ or $t$ according to the sign of $b_v$ with an edge of capacity $\beta |b_v|$ and scale $\beta$ until the additional edges are saturated.}  

In this paper, we introduce a new approach to this problem in undirected graphs.
We maintain a flow that may not quite route $b$ exactly, but we also keep track of an upper bound on how much it will cost us in congestion to fix it back up.  We will aim to minimize a potential function measuring the current congestion plus an over-estimate on the cost of fixing up the residuals.  By not needing to worry about precisely conserving flow at every vertex, we can take large steps in each iteration towards minimizing our potential function.  On the other hand, by intentionally over-estimating the cost of fixing up the residuals, in the course of minimizing our potential function we must inevitably fix them up, as it will cost strictly less to do so.

For a graph $G$, let $C$ be the $m \times m$ diagonal matrix containing the edge capacities, and let $B$ be the $n \times m$ divergence matrix, where $(Bf)_i$ is the excess at vertex $i$.  For a set $S \subseteq V$, we'll write $b_S = \sum_{i \in S}b_S$, the total excess in $S$, and $c_S = \sum_{e: S \leftrightarrow V\setminus S}c_e$, the capacity of the cut $(S,V\setminus S)$ in $G$.  A valid demand vector satisfies $b_V = 0$. 

The minimum congestion flow problem for demands $b$, and its dual, the maximum congested cut, are
\begin{eqnarray}
 \min & \|C^{-1}f\|_\infty  & s.t.\ Bf=b \label{minopt}\\
 \max & b^\top v & s.t.\ \|CB^\top v\|_1 \leq 1 \label{maxopt}
\end{eqnarray}
We refer to the optimum value of these problems as $\opt(b)$.  It is well-known that for problem \eqref{maxopt}, one of the threshold cuts with respect to $v$ achieves $b_S/c_S \geq b^\top v$.  

\subsection{History}
Much of the early work on this problem considers the general, directed edge case, culminating in the still-best binary blocking flow algorithm of Goldberg and Rao\cite{GoldbergRao} that achieves $\tilde{O}(m\min(m^{1/2},n^{2/3})$ time.  Karger and Levine\cite{KargerLevine} give evidence that the undirected case seems easier in graphs with small flow values.  The smooth sparsification technique of Bencz\'{u}r and Karger\cite{Sparsify} shows one can split a graph with $m$ edges into $t= \tilde{O}(m\eps^2/n)$ graphs each with with $\tilde{O}(n\eps^{-2})$ edges, and each of which has at least $(1-\eps)/t$ of the capacity of the original graph.  Therefore, for undirected graphs, any algorithm running in time $T(m,n)$ can be replaced with one running in time $\tilde{O}(m) + (m/n)T(\tilde{O}{n},n)$.  Using the algorithm of Goldberg and Rao yields approximate maximum flows in $\tilde{O}(mn^{1/2})$ time, and for many years that was the best known.

In a breakthrough, Christiano \emph{et al.} show how to compute
approximate maximum flows in $\tilde{O}(mn^{1/3})$ time\cite{Christiano}.  Their new approach uses the nearly linear-time Laplacian solver of Spielman and Teng\cite{ST} to take steps in the minimization of a softmax approximation of the edge congestions.  Each step involves minimizing $\|WC^{-1}f\|_2$, a weighted $\ell_2$-norm of the congestions.  While a naive analysis yields immediately yields a method that makes $\tilde{O}(\sqrt{m})$ such iterations (because $\|\cdot\|_2$ approximates $\|\cdot\|_\infty$ by a $\sqrt{m}$ factor), they present a surprising and insightful analysis showing in fact only $\tilde{O}(m^{1/3})$ such $\ell_2$ iterations suffice.  The maximum-flow-specific parts of \cite{Christiano} are quite simple, needing only to maintain the weights $W$ and then using the Spielman-Teng solver as a black box.

\subsection{Outline}
In this work we pry open that box, extract the parts we need, and apply them \emph{directly} to the maximum flow problem.  In the course of doing so, we push the running time for this decades-old problem nearly down to linear.  Solving a Laplacian system $Lv^*=b^*$, where $L = BCB^\top$, corresponds to finding a flow $f^*$ with $Bf^* = b^*$ minimizing $\|C^{-1/2}f\|_2$.  While Spielman and Teng work entirely in the dual space of vertex potentials and never explicitly represent flows, we can still translate between spaces throughout the algorithm to get an idea of what is actually going on.  At any point, the vertex potentials $v$ induce an optimal flow for \emph{some} demands $b = Lv$; just perhaps not the ones desired.  The solution of Spielman and Teng is to maintain a simpler graph $G'$ that approximates $G$, in the sense that the $\ell_2$ cost of routing flows in $G'$ is within some factor $\alpha$ of the cost in $G$.  The residual flow is routed optimally in $G'$(recursively), by solving $L'v' = b^*-b$.  The potentials that induced that flow are added to $v$, in the hope that $Lv + Lv' \approxeq Lv + L'v' = b^*$.  Indeed, if $G'$ approximates $G$ by a factor $\alpha$, then $b^*-b$ gets smaller in a certain norm by $(1-1/\alpha)$, nudging the flow towards actually routing $b^*$.

Our first step towards obtaining a ST-like algorithm is the definition of a good approximator for the congestion required by $\ell_\infty$-flows.
\begin{defn} An $\alpha$-congestion-approximator for $G$ is a matrix $R$ such that for any demand vector $b$,
\[ \|Rb\|_\infty \leq \opt(b) \leq \alpha\|Rb\|_\infty \]
\end{defn}

Our main result is that we can use good congestion approximators to quickly find
near-optimal flows in a graph.  We prove the following theorem in section \ref{descent}.
\begin{theorem}\label{routethm} There is an algorithm that, given demands $b$ and access to an $\alpha$-congestion-approximator $R$, makes $\tilde{O}(\alpha \eps^{-2}\log^2(n))$ iterations and then returns a flow $f$ and cut $S$ with $Bf=b$ and $\|C^{-1}f\| \leq (1+\eps)b_S/c_S$.  Each iteration requires $O(m)$ time, plus a multiplication by $R$ and $R^\top$.
\end{theorem}

For the sake of giving the reader a concrete example of what a congestion-approximator might look like before continuing, we'll begin with two simple toy examples.
\begin{example} Let $T$ be a maximum weight spanning tree in $G$, and let $R$ be the $n-1 \times n$ matrix with a row for each edge in $T$, and
\[ (Rb)_e = \frac{b_S}{c_S} \]
where $(S,V\setminus S)$ is the cut in $G$ induced by removing $e$ from $T$.

Then, $R$ is a $m$-congestion-approximator.
\end{example} \begin{proof}  Since $(Rb)_e$ is the congestion on the cut in $G$ induced by removing $e$ from $T$, certainly $\opt(b) \geq \|Rb\|_\infty$.  On the other hand, at least $1/m$ of the capacity of those cuts is contained in $T$, so routing $b$ through $T$ congests $e$ by at most $m|(Rb)_e|$.  Multiplication by $R$ and $R^\top$ can be done in $O(n)$ time via elimination on leaves.
\end{proof}

For graphs of large conductance, we can obtain a trivial approximator by simply looking at how much the demand into each vertex congests its total degree.
\begin{example} Let $G$ have conductance $\phi$.  Let $R$ be $n \times n$ diagonal matrix with $R_{i,i} = 1/\deg(i)$ where $\deg(i) = c_{\{i\}}$.  Then, $R$ is a $\phi^{-1}$-congestion-approximator.
\end{example} \begin{proof} Routing $|b_i|$ into or out of vertex $i$ certainly must congest one of its edges by at least $|b_i|/\deg(i)$, so $\opt(b) \geq \|Rb\|_\infty$.  On the other hand, the capacity of any cut in $G$ is at least $\phi$ times the total degree of the smaller side.  It follows that if no vertex is congested by more than $\beta$, then no cut is congested by more than $\phi^{-1}\beta$. \end{proof}

Those two simple examples are analogous to the simple cases for $\ell_2$ flows in the ST-algorithm.  The former is analogous to preconditioning by a small-stretch spanning tree, while the latter is analogous to not preconditioning at all.
As in the ST-algorithm, those two simple examples in fact capture the ideas behind our real constructions.  In section \ref{construct}, we show how to apply the \emph{$j$-tree decomposition} of Madry\cite{Approxcut} (itself based on the ultrasparsifiers of Spielman and Teng\cite{ST})) to obtain good congestion approximators for any graph.
\begin{theorem}\label{constructthm} For any $1 \leq k \leq \log n$, there is an algorithm to construct a data structure representing a $\log(n)^{O(k)}$-congestion-approximator $R$ in time $\tilde{O}(m+n^{1+1/k})$.  Once constructed, $R$ and $R^\top$ can each be applied in time $\tilde{O}(n^{1+1/k})$.
\end{theorem}

Combining theorem \ref{constructthm} where $k = \tilde{\theta}(\sqrt{\log n})$ with theorem \ref{routethm} yields our main result of a nearly-linear time algorithm for minimum congestion flows.
\begin{theorem} There is an algorithm to compute $(1+\eps)$-approximate minimum congestion flows in time $m\eps^{-2}\cdot \exp(\tilde{O}(\sqrt{\log n}))$.
For graphs of conductance $\phi$, the same can be done in time $\tilde{O}(m\phi^{-1})$.
\end{theorem}

\section{\label{descent}Congestion Potential}

The key to our scheme lies transforming problem \eqref{minopt} to an
unconstrained optimization problem, using the congestion-approximator to bound the cost of routing the residual.  To that end, we introduce our potential function.

\begin{equation}
 \min \|C^{-1}f\|_\infty + 2\alpha\|R(b-Bf)\|_\infty \label{newopt}
\end{equation}

We believe the \emph{mere statement} of the potential function \eqref{newopt} to be the most important idea in this paper.  Once \eqref{newopt} has been written down, the remainder of our algorithm is nearly obvious.  We solve problem \eqref{newopt} nearly-optimally by approximating $\|\cdot\|_\infty$ with a softmax.  The softmax is well-approximated by its gradient in the region where steps are taken with $\ell_\infty$-norm $O(1)$.  Since $\|RBC\|_{\infty\to\infty} \leq 1$ for a congestion-approximator, we can take steps on edge $e$ of size $\Omega(\alpha^{-1})c_e$.  We include the details at the end of this section, proving the following theorem.
\begin{theorem}\label{almostthm} There is an algorithm \texttt{AlmostRoute} that given $b$ and $\eps \leq 1/2$, performs $\tilde{O}(\alpha\eps^{-2}\log n)$ iterations and returns a flow $f$ and cut $S$ with,
\[  \|C^{-1}f\|_\infty + 2\alpha\|R(b-Bf)\|_\infty \leq (1+\eps)\frac{b_S}{c_S} \]

Each iteration requires $O(m)$ time plus a multiplication by $R$ and $R^\top$.
\end{theorem}

The flow $f$ may not quite route the demands we wanted.  Fortunately, that will be easy to fix.  The extra factor of two in equation \eqref{newopt} means that half of the contribution to the objective value from the residual part is pure slack.  On the other hand, if $f$ is nearly optimal, there can't be too much slack.
\begin{lemma}\label{slack} Suppose $\|C^{-1}f\|_\infty + 2\alpha\|R(b-Bf)\|_\infty \leq (1+\eps)\opt(b)$.  Then, $\|R(b-Bf)\|_\infty \leq \eps\|Rb\|_\infty$.
\end{lemma}
\begin{proof}
Let $f$ meet the assumption.  Let $f'$ be a routing of $b-Bf$ in $G$ with $\|C^{-1}f'\|_\infty \leq \alpha \|R(b-Bf)\|_\infty$.  Then, moving from $f$ to $f + f'$ decreases the objective value by atleast $\alpha\|R(b-Bf)\|$.  On the other hand, that decrease can't exceed $\eps \opt(b)$.  Since $\opt(b) \leq \alpha\|Rb\|_\infty$, the lemma follows.
\end{proof}

So while $\texttt{AlmostRoute}$ may not route $b$, our bound for the congestion to route the residual is at most half of our bound to route the original demands.  Furthermore, the objective already pays the cost of routing that residual.  In fact, it pays it with a factor of two, so we need only route the remaining residual within a factor-two of optimal.   That suggests an obvious way to route demands $b$: repeatedly invoke $\texttt{AlmostRoute}$ on the remaining residual, until the congestion required to route it is extremely small compared to the congestion required to route $b$.  Then route the final residual in a naive way, such as via a maximal spanning tree.  The cost of that final routing will be paid for by the slack in the objective value of the first routing, simply by finding a factor $3/2$-optimal routing for each residual after the first case. 

Formalizing the latter argument completes our proof of theorem \ref{routethm}.  Set $b_0 \gets b$, and let $(f_0, S_0) \gets \mathtt{AlmostRoute}(b_0, \eps)$.  Next for $i = 1, \ldots, T$ where $T = \log(2m)$, set $b_i \gets b_{i-1}-Bf_{i-1}$ and $(f_i, S_i) \gets \mathtt{AlmostRoute}(b_i, 1/2)$ (we don't actually need any $S_i$ after $S_0)$. Finally, let $b_{T+1} = b_t-Bf_t$, and let $f_{T+1}$ be a flow routing $b_{T+1}$ in a maximal spanning tree of $G$.
Output  $f_1 + \cdots + f_{T+1}$ and $S_0$.  Observe that
theorem \ref{almostthm} yields
\begin{eqnarray*}
 \|C^{-1}f_0\|_\infty + 2\alpha\|Rb_1\|_\infty &\leq& (1+\eps)b_{S_0}/c_{S_0} \\
 \|C^{-1}f_i\|_\infty + 2\alpha\|Rb_{i+1}\|_\infty &\leq &(3/2)\opt(b_i) \leq (3/2)\alpha\|Rb_i\|_\infty 
 \end{eqnarray*}
Beginning with the former inequality and repeatedly applying the latter yields,
\begin{eqnarray*}
(1+\eps)\frac{b_{S_0}}{c_{S_0}} & \geq & \|C^{-1}f_0\|_\infty + 2\alpha\|Rb_1\|_\infty \\
 & \geq & (1/2)\alpha\|Rb_1\|_\infty + \|C^{-1}f_0\|_\infty + \|C^{-1}f_1\|_\infty + 2\alpha\|Rb_2\|_\infty \\
 & \vdots & \\
 & \geq & (1/2)\alpha\|Rb_1\|_\infty + \|C^{-1}f_0\|_\infty + \cdots + \|C^{-1}f_T\|_\infty + 2\alpha\|Rb_{T+1}\|_\infty
\end{eqnarray*}
On the other hand, by choice of $T$, we have
\[ \|C^{-1}f_{T+1}\|_\infty \leq m\alpha\|Rb_{T+1}\|_\infty \leq m\alpha2^{-T}\|Rb_1\|_\infty \leq (1/2)\alpha\|Rb_1\|_\infty \]
Combining the two yields the theorem.

\subsection{Proof of Theorem \ref{almostthm}}

For this preliminary draft, we prove theorem \ref{almostthm} with slightly worse parameters of $\tilde{O}(\alpha^2\eps^{-3}\log^2(n))$, using naive steepest descent.  The better parameters follow from using Nesterov's accelerated gradient method\cite{Nesterov} instead, as will be proved in the final version of this paper.  Of course, for the case of $\alpha = \exp(\tilde{O}(\sqrt{\log n})$, then still $\alpha^2 = \exp(\tilde{O}(\sqrt{\log n}))$, so this naive analysis still yields nearly linear-time flow algorithms.

We approximate $\|\cdot\|_\infty$ using the symmetric softmax function.
\[ \lmax(x) = \log\left(\sum_i e^{x_i} + e^{-x_i}\right) \]

We make use of some elementary facts about $\lmax$.
\begin{fact}\label{lmaxlemma} Let $x,y \in \R^d$.  Then,
\begin{eqnarray}
\|\nabla \lmax(x)\|_1 & \leq & 1 \label{lessone}\\
\nabla \lmax(x)^\top x & \geq & \lmax(x) - \log(2d) \label{probbig} \\
\|\nabla \lmax(x) - \nabla\lmax(y)\|_1 & \leq  &\|x-y\|_\infty \label{lipschitz}
\end{eqnarray}
\end{fact}

We will approximate problem \eqref{newopt} with the potential function,
\begin{equation}
 \phi(f) = \lmax(C^{-1}f) + \lmax\left(2\alpha R(b-Bf)\right) \label{phi}
\end{equation}

Since equation \eqref{phi} approximates equation \eqref{newopt} to within an additive $\theta(\log n)$, we will be concerned with minimizing $\phi(f)$ after scaling $f,b$ so $\phi(f) = \theta(\eps^{-1}\log n)$.

\mybox{
$\mathtt{AlmostRoute}(b,\eps)$:
\begin{itemize}
\item Initialize $f= 0$, scale $b$ so $2\alpha\|Rb\|_\infty = 16\eps^{-1}\log(n)$.
\item Repeat:
\begin{itemize}
\item While $\phi(f) < 16\eps^{-1}\log(n)$, scale $f$ and $b$ up by $17/16$.
\item Set $\delta \gets  \|C\nabla\phi(f)\|_1$.
\item If $\delta \geq \eps/4$, set $f_e \gets f_e - \frac{\delta}{1+4\alpha^2}\sgn(\nabla\phi(f)_e)c_e$
\item Otherwise, terminate and output $f$ together with the potentials induced by $\nabla\phi(f)$(see below), after undoing any scaling.
\end{itemize}
\end{itemize}
}

Each step requires computing $\nabla \phi(f)$, which requires $O(m)$ time plus a multiplication by $R$ and a multiplication by $R^\top$.  Further, the partial derivative of the residual part for a particular edge is equal to a potential difference between the endpoints of that edge.  When $\phi(f)$ is nearly-optimal, those potentials yield a good dual solution for our original problem.

\begin{lemma}\label{optphi} When $\mathtt{AlmostRoute}$ terminates, we have a flow $f$ and potentials $v$ with,
\[ \|C^{-1}f\|_\infty + 2\alpha \|R(b-Bf)\|_\infty \leq (1+\eps)\frac{b^\top v}{\|CB^\top v\|_1} \]
\end{lemma} 
\begin{proof} Set $x_1 = C^{-1}f$, $x_2 = 2\alpha R(b-Bf)$, and $p_i = \nabla \lmax (x_i)$.  Set $v = R^\top p_2$ to be our potentials.  Observe that $\nabla \phi(f) = C^{-1}p_1 -2\alpha B^\top v$.  First, equation \eqref{lessone} yields \[2\alpha\|CB^\top v\|_1 \leq \|p_1\|_1 + \|p_1-2\alpha CB^\top v\|_1 \leq 1 + \delta \]
By equation \eqref{probbig}, using the fact that $C$ and $R$ have at most $n^2/2$ rows and $\phi(f) \geq 16\eps^{-1}\log n$,
\[
 p_1^\top C^{-1}f + 2\alpha p_2^\top R(b-Bf) \geq \phi(f) - 4\log n \geq \phi(f)(1-\eps/4)
\]
On the other hand,
\begin{eqnarray*}
\delta \phi(f) & \geq & \|C\nabla\phi(f)\|_1\|C^{-1}f\|_\infty \\
 & \geq & \nabla\phi(f)^\top f \\
 & = & (C^{-1}p_1-2\alpha B^T R^Tp_2)^\top f \\
 & = & p_1^\top C^{-1}f -2\alpha p_2^\top RBf
 \end{eqnarray*}
Combining the two yields,
\[ 2\alpha b^\top v \geq \phi(f)(1-\eps/4-\delta) \]

Altogether, using the fact that $\delta < \eps/4$ at termination, we have
\[ \frac{b^\top v}{\|CB^\top v\|_1} \geq \frac{\phi(f)(1-\eps/2)}{1+\eps/4} \geq \frac{\phi(f)}{1+\eps} \]
Observing that $\phi(f)$ overestimates $\|C^{-1}f\|_\infty + 2\alpha \|R(b-Bf)\|_\infty$ completes the proof.
\end{proof}

\begin{lemma} $\mathtt{AlmostRoute}$ terminates after at most $\tilde{O}(\alpha^2 \eps^{-3}\log n)$ iterations.
\end{lemma} \begin{proof}  Let us call the iterations between each scaling a \emph{phase}.  Since $\|Rb\|_\infty$ gives us the correct scale to within factor $\alpha$, we will scale at most $O(\log \alpha)$ times.

Let $h_e = -\frac{\delta}{1+4\alpha^2}\sgn(\nabla_f \phi(f)_e)c_e$ be our step.  Then, equation \eqref{lipschitz}, together with the fact that $\|RBC\|_{\infty \to \infty} \leq 1$ for a congestion-approximator $R$ yields,
\begin{eqnarray*}
\phi(f + h) & \leq & \phi(f) + \nabla \phi(f)^\top h + \frac{1+4\alpha^2}{2}\|C^{-1}h\|_\infty^2 \\
 & = & \phi(f) - \frac{\delta^2}{2+8\alpha^2} \\
 & = & \phi(f) - \Omega(\eps^2\alpha^{-2})
\end{eqnarray*}

Since we raised $\phi(f)$ by at most $\eps^{-1}\log n$ when scaling, and each step drops $\phi(f)$ by at least $\Omega(\eps^2\alpha^{-2})$, there can be at most $O(\alpha^2\eps^{-3}\log n)$ steps between phases.
\end{proof}

\section{\label{construct}Computing Congestion-Approximators}

In this section we prove theorem \ref{constructthm}, using a construction of Madry\cite{Approxcut}, itself based on a construction of Spielman and Teng\cite{ST}.
\begin{defn}[Madry\cite{Approxcut}] A $j$-tree is a graph formed by the union of a forest with $j$ components, together with a graph $H$ on $j$ vertices, one from each component.  The graph $H$ is called the core.
\end{defn}

\begin{theorem}[Madry\cite{Approxcut}] \label{madrythm}
For any graph $G$ and $t \geq 1$, we can find in time $\tilde{O}(tm)$ a distribution of $t$ graphs $(\lambda_i, G_i)$ such that,
\begin{itemize}
\item Each $G_i$ is a $O(m \log m/t)$-tree, with a core containing at most $m$ edges.
\item $G_i$ dominates $G$ on all cuts.
\item $\sum_i \lambda_i G_i$ can be routed in $G$ with congestion $\tilde{O}(\log n)$.
\end{itemize}
\end{theorem}

We briefly remark that while the statement of theorem \ref{madrythm} in \cite{Approxcut} contains an additional logarithmic dependence on the \emph{capacity-ratio} of $G$, that dependence is easily eliminated.  We elaborate further in appendix \ref{madryfix}.  Our construction will simply apply theorem \ref{madrythm} recursively, sparsifying the core on each iteration.  To accomplish that,  we use an algorithm of Bencz\'{u}r and Karger\cite{Sparsify}.
\begin{theorem}[Bencz\'{u}r, Karger\cite{Sparsify}]\label{sparsethm}
There is an algorithm $\mathtt{Sparsify}(G,\eps)$ that, given a graph $G$ with $m$ edges, takes $\tilde{O}(m)$ time and returns a graph $G'$ with $m' = O(n\eps^{-2} \log n)$ edges such that the capacity of cuts in the respective graphs satisfy
\[ G \leq G' \leq (1+\eps)G \]
 Further, the edges of $G'$ are scaled versions of a subset of edges in $G$, with no edge scaled by more than $(1+\eps)m/m'$.
\end{theorem}

We now present the algorithm for computing the data structure representing a congestion-approximator.  The algorithm $\mathtt{ComputeTrees}$ assumes its input is sparse; our top-level data-structure is constructed by invoking $\mathtt{ComputeTrees}(\mathtt{Sparsify}(G,1), n^{1/k})$, where $k$ is the parameter of theorem \ref{constructthm}.

\mybox{
$\mathtt{ComputeTrees}(G,t)$:
\begin{itemize}
\item If $n=1$, return.
\item Using theorem \ref{madrythm}, compute distribution $(\lambda_i, G_i)_i^{t'}$ of $\max(1,n/t)$-trees.  
\item Pick the $t$ graphs of largest $\lambda_i$, throw away the rest, and scale the kept $\lambda_i$ to sum to $1$.
\item For $i = 1,\ldots,t$:
\begin{itemize}
\item $H_i' \gets \mathtt{Sparsify}(H_i, 1)$, where $H_i$ is the core of $G_i$
\item $L_i \gets \mathtt{ComputeTrees}(H_i', t)$
\end{itemize}
\item Return the list $L = (\lambda_i, F_i, L_i)_{i=1}^t$ where $F_i$ is the forest of $G_i$.
\end{itemize}
}

The analysis of $\mathtt{ComputeTrees}$ correctness will make use of another algorithm for sampling trees.  The $\mathtt{SampleTree}$ procedure is only used for analysis, and is not part of our flow algorithm.

\mybox{
$\mathtt{SampleTree}(L = (\lambda_i, F_i, T_i)_i^t)$:
\begin{itemize}
\item Pick $i$ with probability $\lambda_i$.
\item Output $F_i+\mathtt{SampleTree}(T_i)$
\end{itemize}
}

\begin{lemma} Let $G$ have $\tilde{O}(n)$ edges, and set $L \gets \mathtt{ComputeTrees}(G, t)$.   Then, every tree in the sample space of $\mathtt{SampleTree}(L)$ dominates $G$ on all cuts, and $\E[\mathtt{SampleTree}(L)]$ is routable in $G$ with congestion $\log(n)^{\log(n)/log(t)}$.  Further, the computation of $L$ takes $\tilde{O}(tn)$ time.
\end{lemma}
\begin{proof} By induction on $n$.  For $n=1$ the claim is vacuous, so suppose $n = t^{k+1}$.  Since $G$ has $O(n \log n)$ edges, the distribution output by theorem \ref{madrythm} will have $O(t \log^2 n)$ entries.  We have $H_i' \geq H_i$ and $H_i + F_i \geq G$.  Furthermore, the inductive hypothesis implies that every tree $T_i$ in $\SampleTree(L_i)$ dominates $H_i'$.  Then,
\[ T_i + F_i \geq H'_i + F_i \geq H_i + F_i  = G_i \geq G \]

Sparsifying the distribution from $O(t\log^2 n)$ to $t$ scales $\lambda_i$ by at most $O(\log^2 n)$, so that $\sum_{i=1}^t\lambda_i G_i$ is routable in $G$ with congestion at most $\log^2 n$ larger than the original distribution.  Since $H_i' \leq 2H_i$, by the multicommodity max-flow/min-cut theorem\cite{LeightonRao} $H_i'$ is routable in $H_i$ with congestion $O(\log n)$.  By the inductive hypothesis, $\E[\SampleTree(L_i)]$ is routable in $H_i'$ with congestion $\log^{O(k)}(n)$.  It follows then that $\E[\SampleTree(L)]$ is routable in $G$ with congestion at most $\log^{O(k+1)}(n)$.

Finally, $\mathtt{ComputeTrees}$ requires $\tilde{O}(tn)$ time to compute the distribution, another $\tilde{O}(tn)$ time to sparsify the cores, and then makes $t$ recursive calls on sparse graphs with $n/t$ vertices.  It follows that the running time of $\mathtt{ComputeTrees}$ is $\tilde{O}(tn)$. 
\end{proof}

\begin{lemma} Let $R$ be the matrix that has a row for each forest edge in our data structure, and $(Rb)_e$ is the congestion on that edge when routing $b$.  If $\E[\SampleTree(L)]$ is routable in $G$ with congestion $\alpha$, then $R$ is a $\alpha$-congestion-approximator for $G$.  Further, $R$ has $\tilde{O}(tn)$ rows.
\end{lemma}
\begin{proof} Since the capacity of each tree-edge dominates the capacity of the corresponding cut in $G$, $\opt(b) \geq \|Rb\|_\infty$.  On the other hand, $b$ can be routed in every tree with congestion $\|Rb\|_\infty$.  By routing a $\Pr[T]$ fraction of the flow through tree $T$, we route $b$ in $\E[\SampleTree(L)]$ with congestion $\|Rb\|_\infty$.  But then $b$ can be routed in $G$ while congesting by at most an $\alpha$ factor larger.

The total number of edges in $R$ satisfies the recurrence $E(n) \leq nt + tE(n/t)$ as each edge is either in one of the $t$ toplevel forests, or in one of the $t$ subgraphs.
\end{proof}

Having constructed our representation of $R$, it remains only to show how to multiply by $R$ and $R^\top$.  We use the following lemmas as subroutines, which are simple applications of leaf-elimination on trees.
\begin{lemma}\label{treeflow} There is an algorithm $\mathtt{TreeFlow}$ that, given a tree $T$ and a demand vector $b$, takes $O(n)$ time and outputs for each tree edge, the flow along that edge when routing $b$ in $T$.
\end{lemma}

\begin{lemma}\label{treepot} There is an algorithm $\mathtt{TreePotential}$ that, given a tree $T$ annotated with a price $p_e$ for each edge, takes $O(n)$ time and outputs a vector of vertex potentials $v$ such that, for any $i,j$, the sum of the prices on the path from $j$ to $i$ in $T$ is $v_i - v_j$.
\end{lemma}

We begin with computing $R$.  We take as input the demand vector $b$, and then annotate each forest edge $e$ with the congestion $r_e$ induced by routing $b$ through a tree containing $e$.

\mybox{
$\mathtt{Compute}R(b, L = (\lambda_i, F_i, T_i)_i^t)$:
\begin{itemize}
\item For $i = 1,\ldots,t$:
\begin{itemize}
\item Let $T$ be the tree formed by taking $F_i$, adding a new vertex $s$, and an edge from $s$ to each core-vertex of $F_i$.  Augment $b$ with demand zero to the new vertex.
\item $f \gets \mathtt{TreeFlow}(b, T)$.
\item Set $r_e \gets f_e/c_e$ for each forest edge in $F_i$.
\item Set $b'$ to a vector indexed by core-vertices, with $b'_j$ equal to the flow on the edge from $s$ to core-vertex $j$.
\item $\mathtt{ComputeR}(b', T_i)$.
\end{itemize}
\end{itemize}
}

\begin{lemma} The procedure $\mathtt{Compute}R(b, L)$ correctly annotates each edge $e$ with $r_e = (Rb)_e$, and takes $\tilde{O}(tn)$ time.
\end{lemma} \begin{proof} Let $L=(\lambda_i,F_i,T_i)_i^t$.  We argue by induction on the depth of recursion. Fix a level and index $i$.  Observe that the cut in $G$ induced by cutting a forest edge is the same regardless of what tree $T$ lies on the core: it is the cut that separates the part of $F_i$ not containing the core from the rest of the vertices.  It follows that we may place \emph{any} tree on the core vertices, invoke $\mathtt{TreeFlow}$, and obtain the flow on each forest edge.  Next, for each component $S$ of $F_i$, the total excess $b_S$ must enter $S$ via the core vertex.  It follows that in a flow routing $b$ on $F_i + T'$, for any tree $T'$, the restriction of that flow to $T'$ must have excess $b_S$ on the core vertex of $S$, so it suffices to find a flow in the core with demands $b'_j = b_{S_j}$.  But routing $b$ in $F_i + T$ will place exactly $b_{S_j}$ units of flow on the edge from $s$ to core-vertex $j$.

The running time consists of $t$ invocations of $\mathtt{TreeFlow}$ each taking $O(n)$ time, plus $t$ recursive calls on graphs of size $n/t$, for a total running time of $\tilde{O}(tn)$.
\end{proof}

To compute $R^\top$, we assume each forest edge $e$ has been annotated with a price $p_e$ that must be paid by any flow per unit of congestion on that edge, and output potentials $v$ such that $v_i - v_j$ is the total price to be paid for routing a unit of flow from $j$ to $i$.

\mybox{
$\mathtt{Compute}R^\top(L = (\lambda_i, F_i, T_i)_i^t)$:
\begin{itemize}
\item $v \gets 0$
\item For $i = 1,\ldots,t$:
\begin{itemize}
\item $v' \gets \mathtt{Compute}R^\top(T_i)$.
\item Let $T$ be the tree formed by taking $F_i$, adding a vertex $s$, and an edge from $s$ to each core-vertex of $F_i$.  Set $q_e = p_e/c_e$ for each forest edge, and $q_e = v'_j$ for edge $e$ from $s$ to core-vertex $j$.
\item $v'' \gets \mathtt{TreePotential}(T, q)$
\item Add $v''$ to $v$ after removing the entry for $s$.
\end{itemize}
\item Return $v$
\end{itemize}
}

\begin{lemma} Given edges annotated with per-congestion prices, the procedure $\mathtt{Compute}R^\top(L)$ correctly returns potentials $v$ such that $v_k - v_j$ is the cost per unit of flow from vertex $j$ to $k$.
\end{lemma} 
\begin{proof} Let $L=(\lambda_i,F_i,T_i)_i^t$.  We argue by induction on the depth of recursion. Fix a level; a flow must pay its toll to each $G_i$, so the resulting potential equals the sum of the potentials for each $i$.  Fix an index $i$.  A unit of flow from $j$ to $k$ is first routed from $j$ to the core-vertex of the component of $F_i$ containing $j$, then to the core-vertex of the component containing $k$, and then finally to $k$.  By induction, we assume that $v'$ yields potentials that give the per-unit costs of routing between core-vertices.  Placing a star on the core with the edge from $s$ to core-vertex $j$ having per-unit cost $v'_j$ preserves those costs.   If $p_e$ is the price of an edge per unit of congestion, then $q_e = p_e/c_e$ is the price of an edge per unit of flow.  It follows that the total toll paid is the same as the toll paid in $T$; thus, the potentials output by $\mathtt{TreePotential}(T, q)$ are correct.

The running time consists of $t$ recursive calls to $\mathtt{Compute}R^\top$ on graphs of size $n/t$, plus $t$ invocations of $\mathtt{TreePotential}$ each taking $O(n)$ time, for a total running time of $\tilde{O}(tn)$.
\end{proof}

\section{Final Remarks}
We remark that there are many other ways to obtain good congestion approximators.  The oblivious routing schemes of \cite{SatishObliv, RackeObliv} require polynomial time to compute, but, once computed, give us a single tree whose single-edge cuts yield a $\log(n)^{O(1)}$-congestion approximator.  Furthermore, we only need the actual tree, and not the routings of the tree back in the original graph.  If such a single tree could be computed in nearly-linear time, it would make an ideal candidate for use in our algorithm.

There have been substantial simplifications to Spielman and Teng's original algorithm (see \cite{Peng, KelnerSDD}).  It may be possible to use some of those techniques to further simplify our algorithm.

\ifFOCS
\else
\bibliographystyle{plain}
\fi
\bibliography{maxflow}

\begin{thebibliography}{10}

\bibitem{Sparsify}
Andr{\'a}s~A. Bencz{\'u}r and David~R. Karger.
\newblock Randomized approximation schemes for cuts and flows in capacitated
  graphs.
\newblock {\em CoRR}, cs.DS/0207078, 2002.

\bibitem{RackeObliv}
Marcin Bienkowski, Miroslaw Korzeniowski, and Harald R{\"a}cke.
\newblock A practical algorithm for constructing oblivious routing schemes.
\newblock In {\em SPAA}, pages 24--33. ACM, 2003.

\bibitem{Christiano}
Paul Christiano, Jonathan~A. Kelner, Aleksander Madry, Daniel~A. Spielman, and
  Shang-Hua Teng.
\newblock Electrical flows, laplacian systems, and faster approximation of
  maximum flow in undirected graphs.
\newblock In Lance Fortnow and Salil~P. Vadhan, editors, {\em STOC}, pages
  273--282. ACM, 2011.

\bibitem{GoldbergRao}
Andrew~V. Goldberg and Satish Rao.
\newblock Beyond the flow decomposition barrier.
\newblock {\em J. ACM}, 45(5):783--797, 1998.

\bibitem{SatishObliv}
Chris Harrelson, Kirsten Hildrum, and Satish Rao.
\newblock A polynomial-time tree decomposition to minimize congestion.
\newblock In {\em Proceedings of the fifteenth annual ACM symposium on Parallel
  algorithms and architectures}, SPAA '03, pages 34--43, New York, NY, USA,
  2003. ACM.

\bibitem{KargerLevine}
David~R. Karger and Matthew~S. Levine.
\newblock Finding maximum flows in undirected graphs seems easier than
  bipartite matching.
\newblock In Jeffrey~Scott Vitter, editor, {\em STOC}, pages 69--78. ACM, 1998.

\bibitem{KelnerSDD}
Jonathan~A. Kelner, Lorenzo Orecchia, Aaron Sidford, and Zeyuan~Allen Zhu.
\newblock A simple, combinatorial algorithm for solving sdd systems in
  nearly-linear time.
\newblock {\em CoRR}, abs/1301.6628, 2013.

\bibitem{Peng}
Ioannis Koutis, Gary~L. Miller, and Richard Peng.
\newblock A fast solver for a class of linear systems.
\newblock {\em Commun. ACM}, 55(10):99--107, 2012.

\bibitem{LeightonRao}
Tom Leighton and Satish Rao.
\newblock Multicommodity max-flow min-cut theorems and their use in designing
  approximation algorithms.
\newblock {\em J. ACM}, 46(6):787--832, 1999.

\bibitem{Approxcut}
Aleksander Madry.
\newblock Fast approximation algorithms for cut-based problems in undirected
  graphs.
\newblock In {\em FOCS}, pages 245--254. IEEE Computer Society, 2010.

\bibitem{Nesterov}
Yu~Nesterov.
\newblock Smooth minimization of non-smooth functions.
\newblock {\em Math. Program.}, 103(1):127--152, May 2005.

\bibitem{ST}
Daniel~A. Spielman and Shang-Hua Teng.
\newblock Nearly-linear time algorithms for preconditioning and solving
  symmetric, diagonally dominant linear systems.
\newblock {\em CoRR}, abs/cs/0607105, 2006.

\end{thebibliography}

\appendix
\section{\label{madryfix}Fixing Theorem \ref{madrythm}}
The proof of theorem \ref{madrythm} maintains a length function $l(e)$ for each edge, and repeatedly invokes an algorithm $\mathtt{SmallStretchTree}$ that returns a spanning tree $T$ on $G$ with,
\begin{equation}
 \sum_e l_T(e)c(e) \leq \tilde{O}(\log n) \sum_e l(e)c(e)\label{ssst}
 \end{equation}
where $l_T(e)$ is the length of the path between $e$'s endpoints in the tree.   Without loss of generality, by scaling, we assume $\sum_e l(e)c(e) = m$.  Let $\chi(e,e') = 1$ if $e'$ is a tree edge that lies on the path in $T$ containing $e$.  Then,
\[ \sum_e l_T(e)c(e) = \sum_{e} c(e)\sum_{e'} chi(e,e') l(e') = \sum_{e'}l(e')\sum_{e}\chi(e,e')c(e) = \sum_{e'} l(e')c_T(e') \]
where $c_T(e')$ is the total capacity of edges routed through $e'$ in $T$.

The dependence on the capacity ratio arises from the fact that there may be many different scales of congestion on the edges of $T$.  The solution is simply to replace $l$ with $l'(e) = \frac{l(e) + c(e)^{-1}}{2}$, a mixture of the original lengths with the inverse capacities.  Constructing a small-stretch tree with respect to $l'$ still satisfies equation \eqref{ssst} with an extra factor of two, but also implies no tree edge is congested by more than $\tilde{O}(m)$.
\end{document}